\documentclass[12pt]{article}

\pdfoutput=1

\usepackage[margin=1in]{geometry}
\usepackage[
  bookmarks=true,
  bookmarksnumbered=true,
  bookmarksopen=true,
  pdfborder={0 0 0},
  breaklinks=true,
  colorlinks=true,
  linkcolor=black,
  citecolor=black,
  filecolor=black,
  urlcolor=black,
]{hyperref}
\usepackage[round]{natbib}
\usepackage{amssymb,amsmath,amsthm}
\usepackage{enumitem}
\usepackage[capitalise]{cleveref}
\usepackage{mathtools} 
\usepackage[normalem]{ulem} 

\usepackage{fnbreak}
\usepackage{tikz}

\usetikzlibrary{shapes.geometric, arrows}
\tikzstyle{internal} = [rectangle, rounded corners, minimum width=1cm, minimum height=1cm,text centered, draw=black]
\tikzstyle{leaf} = [diamond, minimum width=1cm, minimum height=1cm, text width = 1.2cm, text centered, draw=black]
\tikzstyle{arrow} = [thick,->,>=stealth]

\theoremstyle{plain}
\newtheorem{theorem}{Theorem}
\newtheorem{corollary}{Corollary}
\newtheorem{lemma}{Lemma}
\newtheorem{proposition}{Proposition}

\theoremstyle{definition}
\newtheorem{definition}{Definition}
\newtheorem{example}{Example}

\newlist{parts}{enumerate}{1}
\Crefname{partsi}{Part}{Parts}
\setlist[parts,1]{label=\alph*.,ref=\alph*}

\Crefname{figure}{Figure}{Figures} 
\Crefname{equation}{Equation}{Equations} 

\DeclareMathOperator{\poly}{poly}

\newcommand{\eqdef}{\triangleq}
\newcommand{\prefs}{\mathcal{P}}
\newcommand{\matchings}{\mathcal{M}}
\newcommand{\impl}{\mathcal{I}}
\newcommand{\matching}{\mu}

\hypersetup{
  pdfauthor      = {Itai Ashlagi <iashlagi@stanford.edu> and Yannai A. Gonczarowski <yannai@gonch.name>},
  pdftitle       = {Stable matching mechanisms are not obviously strategy-proof},
}

\title{Stable matching mechanisms are \\ not obviously strategy-proof}
\author{Itai Ashlagi \and
Yannai A. Gonczarowski \thanks{First draft: November 2015.  Ashlagi: Management Science \&  Engineering, Stanford University, \emph{email}: \mbox{\href{mailto:iashlagi@stanford.edu}{iashlagi@stanford.edu}}. Gonczarowski: Einstein Institute of Mathematics, Rachel \& Selim Benin School of Computer Science \& Engineering, and the Federmann Center for the Study of Rationality,
The Hebrew University of Jerusalem; and Microsoft Research, \emph{email}: \href{mailto:yannai@gonch.name}{yannai@gonch.name}. This paper greatly benefited from discussions with \mbox{Sophie Bade}, \mbox{Elchanan Ben-Porath}, \mbox{Shengwu Li}, \mbox{Jordi Mass\'{o}}, \mbox{Muriel Niederle}, \mbox{Noam Nisan}, \mbox{Assaf Romm}, and \mbox{Peter Troyan}. \mbox{Itai Ashlagi} is supported by  NSF grant SES-1254768. \mbox{Yannai Gonczarowski} is supported by the Adams Fellowship Program of the Israel Academy of Sciences and Humanities; his work is supported by the European Research Council under the European Community's Seventh Framework Programme (FP7/2007-2013) / ERC grant agreement no.\ 249159, by ISF grant 1435/14 administered by the Israeli Academy of Sciences, and by Israel-USA Bi-national Science Foundation (BSF) grant number 2014389.}}

\date{July, 2018}

\begin{document}
\maketitle

\begin{abstract}
Many two-sided matching markets, from labor markets to school choice programs, use a clearinghouse based on  the applicant-proposing deferred acceptance algorithm, which is well known to be strategy-proof for the applicants. Nonetheless, a growing amount of empirical evidence reveals that applicants misrepresent their preferences when this mechanism is used.
This paper shows that no mechanism that  implements a stable matching is \emph{obviously strategy-proof} for any side of the market, a stronger incentive property  than strategy-proofness that was introduced by \cite{Li2015}. A stable mechanism that is obviously strategy-proof for applicants is introduced for the case in which agents on the other side have acyclical preferences.
\end{abstract}


\vspace{.5em}


\section{Introduction}
\label{section:intro}

A number of labor markets and school admission programs that can be viewed as two-sided matching markets  use centralized mechanisms to match agents on both sides of the market (or agents on one side of the market and objects on the other side of the market).
One important criterion in the design of such mechanisms is stability \citep{Roth2002-economist-engineer}, requiring that no two agents, one from each side of the market, prefer each other over the partners with whom they are matched. Another highly desired property is strategy-proofness, which alleviates agents' incentives to behave strategically.\footnote{See also \citet{pathak2008leveling}, which finds that non-strategy-proof mechanisms favor sophisticated players over more na{\"i}ve players.}

Indeed, many clearinghouses have adopted in recent years the remarkable deferred acceptance (DA) mechanism \citep{Gale-Shapley},\footnote{Examples include the National Resident Matching Program \citep{roth1984evolution}, as well as school choice programs in Boston \citep{aprs2005} and New York \citep{apr2009} \citep[see also][]{as2003}.} which finds a stable matching and is strategy-proof for one side of the market, namely the proposing side in the DA algorithm \citep{Dubins-Freedman}.\footnote{This mechanism is  also approximately strategy-proof for all participants in the market \citep{immorlica2005marriage,kojima2009incentives,akl2013}.}\textsuperscript{,}\footnote{Indeed, removing the incentives to ``game the system" was a key factor in the city of Boston's decision   to replace its school assignment mechanism in 2005 \citep{abdulkadiroglu2006changing}.} Interestingly, although participants are advised that it is in their best interest to state their true preferences, empirical evidence suggests that a significant fraction nonetheless attempt to strategically misreport their true preferences \citep{HMRS2017}; this was observed in experiments \citep{chen-sonmez2006}, in surveys \citep{rees2016suboptimal}, and in the field \citep{HRS2016,SS2017}. This paper asks whether one can implement the deferred acceptance outcome via a mechanism whose description makes its strategy-proofness more apparent. Toward this goal, we adopt the notion of \emph{obvious strategy-proofness}, an incentive property  introduced by \cite{Li2015} that is stronger than strategy-proofness.

\cite{Li2015} formulated the idea that it is ``easier to be convinced'' of the strategy-proofness of some mechanisms over others.
He introduces, and characterizes, the class of \emph{obviously strategy-proof}
mechanisms. He shows that, roughly speaking,
obviously strategy-proof mechanisms are those whose strategy-proofness can be proved even
under a cognitively limited proof model that does not allow for contingent reasoning.\footnote{\label{second-price}For instance, this notion separates sealed-bid second-price auctions from ascending auctions (where  bidders only need to decide at any given moment whether to quit or not) and provides a possible explanation as to why more subjects have been reported to behave insincerely in the former than in the latter \citep{kagel1987information}.}
In his paper, \citeauthor{Li2015} studies whether various well-known auction and assignment mechanisms with attractive revenue or welfare properties for one side of the market
can be implemented in an obviously strategy-proof manner. Whether  one may implement stable matchings in an obviously strategy-proof manner remained an open problem.

For the purpose of this paper, we adopt the Gale and Shapley (1962) one-to-one matching market with men and women to represent two-sided matching markets; our main results naturally extend to many-to-one markets such as labor markets and school choice programs.
When women's preferences over men are perfectly aligned, the unique stable matching may be recovered via serial dictatorship, where men, in their ranked order, choose their partners. In this case, a sequential implementation of such serial dictatorship is obviously strategy-proof. (This follows from \cite{Li2015}, who shows that in a two-sided assignment market with agents and objects, serial dictatorship, when implemented sequentially, is obviously strategy-proof.\footnote{Since, after selecting an object, the agent quits the game, no contingent reasoning is needed in order to verify that she must ask for her favorite unallocated object. However, serial dictatorship (the same strategy-proof social choice rule), when implemented by having each agent simultaneously submit a ranking over all objects in advance, is not obviously strategy-proof. This example and the example in \cref{second-price} both demonstrate that whereas strategy-proofness is a property of the social choice rule,
obvious strategy-proofness is a property of the mechanism implementing the social choice rule.})
Generalizing to allow for weaker forms of alignment of women's preferences, we
show that if women's preferences are acyclical \citep{Ergin2002},\footnote{A preference profile for a woman over men is cyclical if there are
three men $a,b,c$ and two women $x,y$ such that $a \succ_x b \succ_x c \succ_y a$.} then the men-optimal stable matching can be implemented via an obviously strategy-proof mechanism. While the obvious truthfulness of the basic questions that we use to construct this implementation (questions of the form ``do you prefer $x$ the most out of all currently unmatched women?'') draws from the same intuition upon which the serial dictatorship mechanism is based, the questions are considerably more flexible, and the order of the questions  more subtle.

The main finding of this paper  is that for general preferences, no mechanism that implements the men-optimal stable matching (or any other stable matching) is obviously strategy-proof for men.
We first prove this impossibility  in a specifically crafted matching market with 3 women and 3 men, in which women have fixed (cyclical) commonly known preferences
and men have unrestricted private preferences. It is then shown that for the impossibility to hold in any market, it is sufficient for some $3$ women to have this structure of preferences over some $3$ men.
Moreover, the same result holds even if women's preferences are privately known.
An immediate implication of these results is  that in a large market, in which women's preferences are drawn independently and uniformly at random, with high probability no implementation of any stable mechanism  is obviously strategy-proof for all men (or even for most men).
These results  apply to school choice settings even when  schools are not strategic and  have commonly known priorities over students. For example, unless schools' priorities  over students are sufficiently aligned, no mechanism that is stable with respect to students' preferences and schools' priorities is obviously strategy-proof for students.

This paper sheds more light on fundamental
differences between two-sided market mechanisms that aim to implement a two-sided notion such as stability, and closely related two-sided  market mechanisms that aim to implement some efficiency notion for one of the sides of the market.
First, as noted, in assignment markets there exists an obviously strategy-proof ex-post efficient mechanism
(serial dictatorship). Second, a variety of ascending auctions, from familiar multi-item auctions
\citep{demange1986multi}  to recently proposed clock auctions
\citep{milgrom2014deferred}, maximize welfare or revenue and are obviously strategy-proof,
despite the latter's being based on  deferred acceptance principles. In contrast, this paper shows that
there is no way to achieve stability that is obviously strategy-proof for either side of the market.

Obvious strategy-proofness was introduced by \cite{Li2015}, who studies this property extensively  in mechanisms with monetary transfers. In settings without transfers, \cite{Li2015} studies this property in implementations of serial dictatorship and top trading cycles. Several papers further study this property in different settings. Closely related is \cite{troyan2016}, who studies  two-sided markets with agents and objects and asks for which priorities for objects one can implement in an obviously strategy-proof manner the Pareto-efficient top trading cycles algorithm. \cite{pt2016} characterize general obviously strategy-proof mechanisms without transfers under a ``richness'' assumption on the preferences domain, and characterize the sequential version of random serial dictatorship under such an assumption via a natural set of axioms that includes obvious strategy-proofness. \cite{bg2016} constructively characterize Pareto-efficient social choice rules that admit obviously strategy-proof implementations in popular domains (object assignment, single-peaked preferences, and combinatorial auctions).
It is worth noting that all three of these papers utilize machinery and observations that originated in this paper.

The paper is organized as follows. \cref{sec:model} provides the model and background, including the definition of obvious strategy-proofness in matching markets.
\cref{sec:special-cases} presents special  cases for which an obviously strategy-proof implementation of the men-optimal stable matching exists.
\cref{sec:impossibility} provides the main impossibility result. \cref{sec:two-sided} presents corollaries in a model where women also have private preferences. \cref{sec:conclusion} concludes.

\section{Preliminaries}
\label{sec:model}

\subsection{Two-sided matching with one strategic side}\label{one-sided}

For the bulk of our analysis it will  be sufficient to consider two-sided markets in which only one side of the market is strategic. We begin by defining the notions of matching and strategy-proofness in such markets.

In a two-sided matching market, the participants are partitioned into a finite set of \emph{men}~$M$ and a finite set of \emph{women}~$W$.
A \emph{preference list} (for some man $m$) over $W$ is a totally ordered subset of $W$ (if some woman $w$ does not appear on the preference list, we think of her as being unacceptable to $m$).
Denote the set of all preference lists over $W$ by $\prefs(W)$.
A \emph{preference profile} $\bar{p}=(p_m)_{m\in M}$ for $M$ over $W$ is a specification of a preference list $p_m$ over~$W$ for each man $m\in M$. (So the set of all preference profiles for $M$ over $W$ is $\prefs(W)^M$.)
Given a preference list $p_m$ for some man $m$,
we write $w\succ_m w'$ to denote that man~$m$ strictly prefers woman $w$ over woman $w'$, (i.e., either woman $w$ is ranked higher than~$w'$ on $m$'s preference list, or $w$ appears on this list while $w'$ does not),
and write $w\succeq_m w'$ if it is not the case that $w' \succ_m w$.

A \emph{matching} between $M$ and~$W$ is a one-to-one mapping between a subset of $M$ and a subset of $W$.
Denote the set of all matchings between $M$ and $W$ by $\matchings$. Given a matching~$\matching$ between $M$ and $W$, for a participant $a\in M\cup W$ we write $\matching_a$ to denote $a$'s match in  $\matching$, or write $\matching_a=a$ if $a$ is unmatched.

A (one-side-querying) \emph{matching rule} is a function $C:\prefs(W)^M\rightarrow\matchings$, from preference profiles for $M$ over $W$ to matchings between $M$ and $W$.

A matching rule $C$ is said to be \emph{strategy-proof} for a man $m$ if for every preference profile $\bar{p}=(p_m)_{m\in M}\in\prefs(W)^M$ and for every (alternate) preference list $p'_m\in \prefs(W)$, it is the case that $C_m(\bar{p}) \succeq_m C_m(p_m',\bar{p}_{-m})$ according to $p_m$.\footnote{As is customary, $(p_m',\bar{p}_{-m})$ denotes the preference profile obtained from $\bar{p}$ by setting the preference list of $m$ to be $p'_m$.}
$C$ is said to be \emph{strategy-proof} if it is strategy-proof for every man.

\subsection{Obvious strategy-proofness}\label{osp}

This section briefly describes the notion of obvious strategy-proofness, developed in great generality by \cite{Li2015}. We rephrase these notions for the special case of deterministic matching mechanisms with finite preference and outcome sets. For ease of presentation, attention is restricted to mechanisms under perfect information; however, the results in this
paper still hold (\emph{mutatis mutandis}) via the same proofs for the general definitions of \cite{Li2015}.\footnote{Readers who are familiar with the general definitions of \cite{Li2015} may easily verify that if a randomized stable obviously strategy-proof (OSP) mechanism exists, then derandomizing it by fixing in advance each choice of nature to some choice made with positive probability yields a deterministic stable OSP mechanism. Furthermore, if some stable mechanism is OSP under partial information, then it is also OSP under perfect information.}

Whereas strategy-proofness is a property of a given matching rule, obvious strategy-proofness is a property of a specific implementation, via a specific mechanism, of such a matching rule. A mechanism implements a matching rule by specifying, roughly speaking, an extensive-form game tree that implements the standard-form game associated (where strategies coincide with preference lists) with the matching rule, where each action at each node of the extensive-form game tree corresponds to some set of possible preference lists for the acting participant. We now formalize this definition.

\begin{definition}[matching mechanism]
A (one-side-querying extensive-form) \emph{matching mechanism} for $M$ over $W$ consists of:
\begin{enumerate}
\item
A rooted tree $T$. The nodes/vertices of the tree are denoted by $V(T)$. The edges of the tree are denoted by $E(T)$ and are directed away from the root: if an edge $e$ is incident with a node $n$ but is not on the path from the root of the tree to $n$, then $e$ is outgoing from $n$. The leaves (nodes with no outgoing edges) of the tree are denoted by $L(T)\subset V(T)$.
\item
A map $X:L(T)\rightarrow\matchings$ from the leaves of $T$ to matchings between $M$ and $W$.
\item
A map $Q:V(T)\setminus L(T)\rightarrow M$, from internal nodes of $T$ to $M$.
\item
A map $A:E(T)\rightarrow 2^{\prefs(W)}$, from edges of $T$ to predicates over $\prefs(W)$, such that all of the following hold:
\begin{itemize}
\item
Each such predicate must match at least one element in $\prefs(W)$.
\item
The predicates corresponding to edges outgoing from the same node are disjoint.
\item
The disjunction (i.e., set union) of all predicates corresponding to edges outgoing from a node $n$ equals the predicate corresponding to the last edge outgoing from a node labeled $Q(n)$ along the path from the root to $n$, or to the predicate matching all elements of $\prefs(W)$ if no such edge exists.
\end{itemize}
\end{enumerate}
\end{definition}

A preference profile $\bar{p}\in\prefs(W)^M$  is said to \emph{pass through} a node $n \in V(T)$ if, for each edge~$e$ along the path from the root of $T$ to $n$, it is the case that $p_{Q(n')}\in A(e)$, where $n'$ is the source node of $e$. That is, the nodes through which $\bar{p}$ passes are the nodes of the path that starts from the root of $T$ and follows, from each internal node $n'$ that it reaches, the unique outgoing edge whose predicate matches the preference list of $Q(n')$.

\begin{definition}[implemented matching rule]
Given an extensive-form matching mechanism~$\impl$, we denote by $C^{\impl}$, called the matching rule \emph{implemented by} $\impl$, the (one-side-querying) matching rule  mapping a preference profile $\bar{p}\in\prefs(W)^M$ to the matching $X(n)$, where $n$ is the unique leaf through which $\bar{p}$ passes. Equivalently, $n$ is the node in $T$ obtained by traversing $T$ from its root, and from each internal node $n'$ that is reached, following the unique outgoing edge whose predicate matches the preference list of $Q(n')$.
\end{definition}

Two preference lists $p,p'\in\prefs(W)$ are said to \emph{diverge} at a node $n\in V(T)$ if there exist two distinct edges~$e,e'$ outgoing from $n$ such that  $p\in A(e)$ and $p'\in A(e')$.

\begin{definition}[obvious strategy-proofness (OSP)]
Let $\impl$ be an extensive-form matching mechanism.
\begin{enumerate}
\item
$\impl$ is said to be \emph{obviously strategy-proof (OSP) for a man $m\in M$} if for every node $n$ with $Q(n)=m$ and for every
$\bar{p}=(p_{m'})_{m'\in M}\in\prefs(W)^M$ and $\bar{p}'=(p'_{m'})_{m'\in M}\in\prefs(W)^M$ that both pass through $n$ such that $p_m$ and $p'_m$ diverge at $n$, it is the case that $C^{\impl}_m(\bar{p}) \succeq_m C^{\impl}_m(\bar{p}')$ according to~$p_m$. In other words, the worst possible outcome for $m$ when acting truthfully (i.e., according to $p_m$) at $n$ is no worse than the best possible outcome for $m$ when misrepresenting his preference list to be $p'_m$ at $n$.
\item
$\impl$ is said to be \emph{obviously strategy-proof (OSP)} if it is obviously strategy-proof for every man $m\in M$.
\end{enumerate}
\end{definition}

\cite{Li2015} shows that obviously strategy-proof mechanisms are, in a precise sense, mechanisms that can shown to implement strategy-proof rules under a cognitively limited proof model that does not allow for contingent reasoning.
To observe how strategy-proofness of the matching rule $C^\impl$ for a man $m\in M$  is indeed a weaker condition than obvious strategy-proofness of the mechanism $\impl$ for $m$, note that the matching rule $C^{\impl}$ is strategy-proof for $m$ if and only if for every node $n$ with $Q(n)=m$ and for every $\bar{p}=(p_m)_{m\in M}\in\prefs(W)^M$ that passes through~$n$ and for every $p'_m\in\prefs(W)$ that diverges from $p_m$ at $n$,\footnote{These conditions imply that $(p'_m,\bar{p}_{-m})$ also passes through $n$.} it is the case that $C^{\impl}_m(\bar{p})\succeq_m C^{\impl}_m(p'_m,\bar{p}_{-m})$ according to $p_m$.\footnote{We emphasize that this rephrased definition is equivalent to the definition of strategy-proofness of the matching rule~$C^{\impl}$ that is given in \cref{one-sided}, however it is not equivalent to standard definition of strategy-proofness of the extensive-form game underlying the mechanism~$\impl$, which would allow each man to condition the type he is ``pretending to be'' under any strategy on the information revealed by other men in preceding nodes. Once we move to the realm of obvious strategy-proofness, the restriction on each strategy to always consistently ``pretend to be'' of the same type is inconsequential, as the definition of OSP considers the case in which other men may play different types when the man in question acts truthfully or deviates. It is for this reason that we have chosen to implicitly define a strategy in the extensive-form game underlying $\impl$ to be restricted to consistently ``pretending to be'' of the same type. This somewhat nonstandard implicit definition of a strategy considerably simplifies notation throughout this paper (by considering only consistent behavior on behalf of every agent) without changing the mathematical meaning of obvious strategy-proofness (or of strategy-proofness of a matching rule) and without limiting the generality of our results.}

\begin{definition}[OSP-implementability]
A (one-side-querying) matching rule $C:\prefs(W)^M\rightarrow\matchings$ is said to be \emph{OSP-implementable} if $C=C^{\impl}$ for some obviously strategy-proof matching mechanism~$\impl$. In this case, we say that \emph{$\impl$ OSP-implements $C$}.
\end{definition}

\subsection{Stability}

We proceed to describe a simplified version of  stability in matching markets as introduced by \cite{Gale-Shapley}.
While, as stated in \cref{one-sided}, for the bulk of our analysis it is sufficient to consider markets in which only men are strategic, to define the notion of stability one must consider not only preferences for the (strategic) men, but also preferences (sometimes called priorities) for the (nonstrategic) women. Women's preference lists and preference profiles are defined analogously with those of men.
We continue to denote a preference profile for men by $\bar{p}=(p_m)_{m\in M}\in\prefs(W)^M$, while denoting a preference profile for women by  $\bar{q}=(q_w)_{m\in M}\in\prefs(M)^W$.

Let $\bar{p}$  and $\bar{q}$ be preference profiles of men and women respectively. A matching $\matching$ is said to be \emph{unstable} with respect to  $\bar{p}$  and $\bar{q}$ if there exist a man $m$ and a woman $w$ each preferring the other over the partner matched to them by $\matching$, or if some participant $a\in M\cup W$ is matched with some other participant not on $a$'s preference list. A matching that is not unstable  is said to be \emph{stable}.
\cite{Gale-Shapley} showed that a stable matching exists with respect to every pair of preference profiles and, furthermore, that for every pair of preference profiles there exists an \emph{$M$-optimal stable matching}, i.e., a stable matching such that each man weakly prefers his match  in this stable matching over his match  in any other stable matching.

We now relate the concept of stability to the (one-side-querying) matching rules and mechanisms defined in the previous sections.
Let $\bar{q}\in\prefs(M)^W$ be a preference profile for $W$ over $M$.
A (one-side-querying) matching rule $C$ is said to be \emph{$\bar{q}$-stable} if for every preference profile $\bar{p}\in\prefs(W)^M$ for $M$ over $W$, the matching $C(\bar{p})$ is stable with respect to $\bar{p}$ and $\bar{q}$. A (one-side-querying) matching mechanism is said to be \emph{$\bar{q}$-stable} if the matching rule that it implements is $\bar{q}$-stable.

We denote by $C^{\bar{q}}:\prefs(W)^M\rightarrow\matchings$ the \emph{$M$-optimal stable matching rule}, i.e., the (one-side-querying, $\bar{q}$-stable) matching rule mapping each preference profile for men $\bar{p}$  to the $M$-optimal stable matching with respect to $\bar{p}$ and $\bar{q}$.
It is well known that $C^{\bar{q}}$
is strategy-proof for all men \citep{Dubins-Freedman}.
Moreover, no other matching rule is  strategy-proof for all men \citep{Gale-Sotomayor-1985}.\footnote{For a more general result, see \cite{chen2016manipulability}.}
In the notation of this paper:

\begin{theorem}[\citealp{Gale-Sotomayor-1985,chen2016manipulability}]\label{strategy-proof}
For every preference profile $\bar{q}\in\prefs(M)^W$ for $W$ over $M$, no $\bar{q}$-stable matching rule $C\ne C^{\bar{q}}$ is strategy-proof.
\end{theorem}

In this paper, we ask whether $C^{\bar{q}}$ is not only strategy-proof, but also OSP-implementable. (As it is the unique strategy-proof $\bar{q}$-stable matching rule, it is the only candidate for OSP-implementability.)

\section{OSP-implementable special cases}
\label{sec:special-cases}

Before stating our main impossibility result, we first present a few special cases in which $C^{\bar{q}}$, the $M$-optimal stable matching rule for a fixed women's preference profile $\bar{q}$, is in fact OSP-implementable. These are the first known OSP mechanisms without transfers that are not dictatorial.\footnote{All OSP mechanisms that are surveyed in the end of the introduction are based upon the query structure of the mechanisms of this \lcnamecref{sec:special-cases}.}

For simplicity, we describe all of these cases under the assumption that the market is balanced (i.e., that $|W|=|M|$) and that all preference lists are full (i.e., that each participant prefers being matched to anyone over being unmatched); generalizing each of the below cases for unbalanced markets or for preference lists for men that are not full is straightforward.\footnote{Indeed, asking any man whether he prefers being unmatched over being matched with any (remaining not-yet-matched) woman never violates obvious strategy-proofness.}
The first case we consider is that in which women's preferences are perfectly aligned.

\begin{example}[$C^{\bar{q}}$ is OSP-implementable when women's preferences are perfectly aligned]\label{perfectly-aligned}
Let $q\in\prefs(M)$ and let $\bar{q}=(q)_{w\in W}$ be the preference profile for $W$ over $M$ in which all women share the same preference list $q$. $C^{\bar{q}}$ is OSP-implementable by the following serial dictatorship mechanism: ask the man most preferred according to $\bar{q}$ which woman he prefers most, and assign that woman to this man (in all leaves of the subtree corresponding to this response), ask the man second-most preferred according to $\bar{q}$ which woman he prefers most out of those not yet assigned to any man, and assign that woman to this man (in all leaves of the subtree corresponding to this response), etc. This mechanism can be shown to be OSP by the same reasoning that \cite{Li2015} uses to show that serial dictatorship is OSP.
\end{example}

Another noteworthy example is that of arbitrary preferences in a very small matching market.

\begin{example}[$C^{\bar{q}}$ is OSP-implementable when $|M|=|W|=2$]\label{two-by-two}
When $|M|=|W|=2$, $C^{\bar{q}}$ is OSP-implementable for every preference profile $\bar{q}\in\prefs(M)^W$ for $W$ over $M$. Indeed, let $M=\{a,b\}$ and $W=\{x,y\}$. If $q_x=q_y$, then $C^{\bar{q}}$ is OSP-implementable as explained in \cref{perfectly-aligned}. Otherwise, without loss of generality $a\succ_x b$ and $b\succ_y a$; for this case, \cref{two-by-two-tree} describes an OSP mechanism that implements $C^{\bar{q}}$.
\begin{figure}[ht]
\centering
\begin{tikzpicture}[node distance=3.5cm]
\node (q1) [internal] {$a$};
\node (q2) [internal, right of=q1] {$b$};
\node (a1) [leaf, below of=q1] {$a \Leftarrow x$ $b \Leftarrow y$};
\node (a2) [leaf, below of=q2] {$a \Leftarrow x$ $b \Leftarrow y$};
\node (a3) [leaf, right of=q2] {$a \Leftarrow y$ $b \Leftarrow x$};
\draw [arrow] (q1) -- node[anchor=west] {$x \succ_a y$} (a1);
\draw [arrow] (q1) -- node[anchor=north] {$y \succ_a x$} (q2);
\draw [arrow] (q2) -- node[anchor=west] {$y \succ_b x$} (a2);
\draw [arrow] (q2) -- node[anchor=north] {$x \succ_b y$} (a3);
\end{tikzpicture}
\caption{An OSP mechanism that implements $C^{\bar{q}}$ for $|W|=|M|=2$ and for $\bar{q}$ where $a\succ_x b$ and $b\succ_y a$. (The notation, e.g., $a\Leftarrow x$, indicates that $x$ is matched to $a$ in the matching corresponding to that leaf of the mechanism tree.)}\label{two-by-two-tree}
\end{figure}
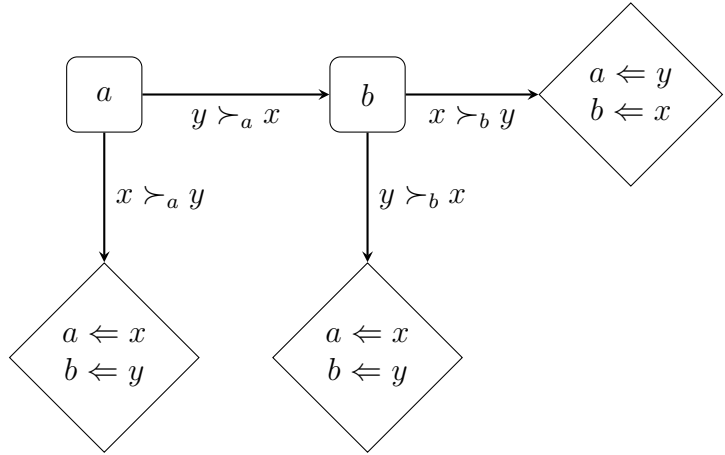
\end{example}

The preference profiles in \cref{perfectly-aligned,two-by-two} are special cases of the class of acyclical preference profiles, whose structure was defined by \cite{Ergin2002}.

\begin{definition}[acyclicality]\label{acyclicality}
A preference profile $\bar{q}\in\prefs(M)^W$ for $W$ over $M$ is said to be \emph{cyclical} if there exist $a,b,c\in M$ and $x,y\in W$ such that $a \succ_x b \succ_x c \succ_y a$.
If $\bar{q}$ is not cyclical, then it is said to be \emph{acyclical}.
\end{definition}

\cite{Ergin2002} shows that acyclicality of $\bar{q}$ is necessary and sufficient for $C^{\bar{q}}$ to be strongly group strategy-proof and Pareto efficient.
We now generalize \cref{perfectly-aligned,two-by-two} by showing that acyclicality of $\bar{q}$ (as in both of these \lcnamecrefs{perfectly-aligned}) is sufficient for $C^{\bar{q}}$ to be also OSP-implementable. Much like the implementations in \cref{perfectly-aligned,two-by-two}, the strategy-proofness of the OSP implementation that emerges for acyclical preferences is far easier to understand than that of the standard deferred-acceptance implementation, thus showcasing the usefulness of obvious strategy-proofness in identifying easy-to-understand implementations. In each mechanism step, either a single man is given free pick out of all remaining $w\in W$, or two men are each given first priority over some subset of $W$ (i.e., free pick if his favorite remaining $w\in W$ is there), and second priority over the rest (i.e., free pick out of all other remaining $w\in W$ except the one chosen by the other man if the latter invoked his first priority).

\begin{theorem}[positive result for acyclical preferences]\label{acyclical-osp}
$C^{\bar{q}}$ is OSP-implementable for every acyclical preference profile $\bar{q}\in\prefs(M)^W$ for $W$ over $M$.
\end{theorem}

\begin{proof}
We prove the result by induction over $|M|=|W|$. By acyclicality, at most two men are ranked by some woman as her top choice. If only one such man $m\in M$ exists, then he is ranked by all women as their top choice---in this case, similarly to \cref{perfectly-aligned}, we ask this man for his top choice $w \in W$, assign her to him, and then continue by induction (finding in an OSP manner the $M$-optimal stable matching between $M\setminus\{m\}$ and $W\setminus\{w\}$). Otherwise, there are precisely two men $a\in M$ and $b\in M$ who are ranked by some woman as her top choice. By acyclicality, each woman either has $a$ as her top choice and $b$ as her second-best choice, or \emph{vice versa}.\footnote{This is reminiscent of the priorities of the first two agents in bipolar serially dictatorial rules \citep{BDE05}, which are indeed included in the analysis of \cref{acyclical-osp} as a special case.} We conclude somewhat similarly to \cref{two-by-two-tree}: for each woman $w\in W$ that prefers $a$ most, we ask $a$ whether he prefers $w$ most; if so, we assign $w$ to~$a$ and continue by induction. Otherwise, for each woman $w\in W$ that prefers $b$ most, we ask $b$ whether he prefers $w$ most; if so, we assign $w$ to $b$ and continue by induction. Otherwise, we ask each of $a$ and $b$ for his top choice, assign each of them his top choice, and continue by induction.

To see that this implementation is OSP, consider a man $m\in M$ who is asked by this mechanism whether a woman $w\in W$ is his top choice (among the remaining women). If $m$ really does prefer $w$ most, then answering truthfully matches him to $w$, which he weakly prefers over any outcome that occurs if he is not truthful. Similarly, if $m$ does not prefer~$w$ most, then answering truthfully may get $m$ a more preferred choice, but also assures $m$ that if he does not get such a preferred choice, then he would still be able to choose to get matched to $w$ (he would do so if he fails to get his top choice, and $w$ is his second-best); so, any outcome that results from truthfulness is weakly preferred by $m$ over any outcome that results from nontruthfulness in this case as well.
\end{proof}

We conclude this section by noting, however, that acyclicality of $\bar{q}$ is not a necessary condition for OSP-implementability of $C^{\bar{q}}$, as demonstrated by the following example.

\begin{example}[OSP-implementable $C^{\bar{q}}$ with cyclical $\bar{q}$]\label{osp-more-cyclical}
Let $M=\{a,b,c\}$ and $W=\{x,y,z\}$. We claim that $C^{\bar{q}}$, for the following cyclical preference profile $\bar{q}$ for $W$ over $M$ (where each woman prefers being matched to any man over being unmatched), is OSP-implementable:
\[\begin{array}{ccccc}
a & \succ_x & b & \succ_x&  c \\
a & \succ_y & c & \succ_y & b \\
b & \succ_z & a & \succ_z & c.
\end{array}\]
We begin by noting that $\bar{q}$ is indeed cyclical, as $a \succ_y c \succ_y b \succ_z a$. We now note that the following mechanism OSP-implements $C^{\bar{q}}$:
\begin{enumerate}
\item
Ask $a$ whether he prefers $x$ the most; if so, assign $x$ to $a$ and continue as in \cref{two-by-two} (finding in an OSP manner the $M$-optimal stable matching between $\{y,z\}$ and $\{b,c\}$).
\item
Ask $a$ whether he prefers $y$ the most; if so, assign $y$ to $a$ and continue as in \cref{two-by-two}. (Otherwise, we deduce that\ \ 1) $a$ prefers $z$ the most and therefore\ \ 2) $c$ will not end up being matched to $z$.)
\item
Ask $b$ whether he prefers $z$ the most; if so, assign $z$ to $b$ and continue as in \cref{two-by-two}.
\item
Ask $b$ whether he prefers $x$ the most; if so, assign $x$ to $b$, $z$ to $a$, and $y$ to $c$. (Otherwise, we deduce that $b$ prefers $y$ the most.)
\item
Ask $c$ whether he prefers $x$ over $y$. If so, assign $x$ to $c$, $y$ to $b$, and $z$ to $a$. (Otherwise, we deduce that $b$ will not end up being matched to $y$.)
\item
Ask $b$ whether he prefers $z$ over $x$. Assign $b$ to his preferred choice between $z$ and $x$ and continue as in \cref{two-by-two}.
\end{enumerate}
\end{example}

Nonetheless, as we show in the next section, when there are more than $2$ participants on each side and women's preferences are sufficiently unaligned, $C^{\bar{q}}$ is not OSP-implementable.

\section{Impossibility result for general preferences}
\label{sec:impossibility}

We now present our main impossibility result.

\begin{theorem}[impossibility result for general preferences]\label{not-osp}
If $|M|\ge3$ and $|W|\ge3$, then there exists a preference profile $\bar{q}\in\prefs(M)^W$ for $W$ over $M$, such that
no $\bar{q}$-stable (one-side-querying) matching rule is OSP-implementable.
\end{theorem}

Observe that \cref{not-osp} applies to any $\bar{q}$-stable (one-side-querying) matching rule, and not only to the $M$-optimal stable matching rule $C^{\bar{q}}$.
Before proving the result, we first prove a special case that cleanly demonstrates the construction underlying our proof.

\begin{lemma}\label{three-by-three}
For $|M|=|W|=3$, there exists a preference profile $\bar{q}\in\prefs(M)^W$ for $W$ over $M$ such that
no $\bar{q}$-stable (one-side-querying) matching rule is OSP-implementable.
\end{lemma}

\begin{proof}
Let $M=\{a,b,c\}$ and $W=\{x,y,z\}$. Let $\bar{q}$ be the following preference profile (where each woman prefers being matched to any man over being unmatched):
\begin{equation}\label{tricyclical}
\begin{array}{ccccc}
a & \succ_x & b & \succ_x&  c \\
b & \succ_y & c & \succ_y & a \\
c & \succ_z & a & \succ_z & b.
\end{array}
\end{equation}
Assume for contradiction that an OSP mechanism $\impl$ that implements a $\bar{q}$-stable matching rule $C^{\impl}$ exists. Therefore, $C^{\impl}$ is strategy-proof, and so, by \cref{strategy-proof}, $C^{\impl}=C^{\bar{q}}$.
In order to reach a contradiction by showing that such a mechanism (that OSP-implements $C^{\bar{q}}$) cannot possibly exist, we dramatically restrict the domain of preferences of all men, which results in a simpler mechanism, where the contradiction can be identified in a less cumbersome manner. We define:

\begin{center}
\setlength\tabcolsep{2em}
\begin{tabular}{ccc}
$p^1_a \eqdef z \succ y \succ x$ & $p^1_b \eqdef x \succ z \succ y$ & $p^1_c \eqdef y \succ x \succ z$ \\
$p^2_a \eqdef y \succ x \succ z$ & $p^2_b \eqdef z \succ y \succ x$ & $p^2_c \eqdef x \succ z \succ y$,
\end{tabular}
\end{center}
and set $\prefs_a\eqdef\{p^1_a,p^2_a\}$, $\prefs_b\eqdef\{p^1_b,p^2_b\}$, and $\prefs_c\eqdef\{p^1_c,p^2_c\}$.

Following the ``pruning'' technique in \cite{Li2015}, we note that if we ``prune'' the tree of $\impl$ by replacing, for each edge $e$, the predicate $A(e)$ with the conjunction (i.e., set intersection) of $A(e)$ with the predicate matching all elements of $\prefs_{Q(n)}$, where $n$ is the source node of~$e$, and by consequently deleting all edges $e$ for which $A(e)=\bot$,\footnote{The standard notation $\bot$ stands for ``false'' (mnemonic: an upside-down ``true'' $\top$), i.e., the predicate that matches nothing, so an edge for which $A(e)=\bot$ will never be followed.} we obtain, in a precise sense, a mechanism that implements $C^{\bar{q}}$ where the preference list of each man $m\in M$ is \emph{a priori} restricted to be in $\prefs_m$.\footnote{The definition of mechanisms and OSP when the domain of preferences is restricted extends naturally from that given in \cref{osp} for unrestricted preferences.
The interested reader is referred to Appendix \ref{restricted-domain-mechanisms} for precise details.}
By a proposition in \cite{Li2015}, since the original mechanism $\impl$ is OSP, so is the pruned mechanism as well.

Let $n$ be the earliest (i.e., closest to the root) node in the pruned tree that has more than one outgoing edge (such a node clearly exists, since $C^{\impl}=C^{\bar{q}}$ is not constant over $\prefs_a\times\prefs_b\times\prefs_c$). By symmetry of $\bar{q},\prefs_a,\prefs_b,\prefs_c$, without loss of generality $Q(n)=a$. By definition of pruning, it must be the case that $n$ has two outgoing edges, one labeled $p^1_a$, and the other labeled $p^2_a$. We claim that the mechanism of the pruned tree is in fact not OSP. Indeed, for $p_a=p^2_a$ (the ``true preferences''), $p_b=p^2_b$, and $p_c=p^1_c$, we have that $C_a^{\impl}(\bar{p})=C_a^{\bar{q}}(\bar{p})=x$, yet for $p'_a=p^1_a$ (a ``possible manipulation''), $p'_b=p^1_b$, and $p'_c=p^2_c$, we have that $C_a^{\impl}(\bar{p}')=C_a^{\bar{q}}(\bar{p}')=y$, even though $C_a^{\impl}(\bar{p}')=y\succ_a x = C_a^{\impl}(\bar{p})$ according to $p_a$ (by definition of $n$, both $\bar{p}$ and $\bar{p}'$ pass through $n$, and $p_a$ and $p_a'$ diverge at $n$), and so the mechanism of the pruned tree indeed is not OSP --- a contradiction.
\end{proof}

\begin{proof}[Proof of \cref{not-osp}]
The \lcnamecref{not-osp} follows from a reduction to \cref{three-by-three}. Indeed, let $a,b,c$ be three distinct men and let $x,y,z$ be three distinct women. Let $\bar{q}\in\prefs(W)^M$ be a preference profile such that the preferences of $x,y,z$ satisfy \cref{tricyclical} with respect to $a,b,c$ (with arbitrary preferences over all other men), and with arbitrary preferences for all other women. Assume for contradiction that a $\bar{q}$-stable OSP mechanism $\impl$ exists.

We prune (see the proof of \cref{three-by-three} for an explanation of pruning) the tree of $\impl$ such that the only possible preference lists for $a,b,c$ are those in which they prefer each of $x,y,z,$ over all other women, and the only possible preference list for all other men is empty.\footnote{Alternatively, one could set for all other men arbitrary preference lists that do not contain $x,y,z$.} Let $\bar{q}'$ be the preference profile given in \cref{three-by-three}; the resulting (pruned) mechanism is a $\bar{q}'$-stable matching mechanism for $a,b,c$ over $x,y,z$,\footnote{Formally, it is a matching mechanism for $W$ over $M$ with respect to the pruned preferences, but can be shown to always leave all participants but $a,b,c$ and $x,y,z,$ unmatched, and so can be thought of as a matching mechanism for $a,b,c$ over $x,y,z$.} and so, by \cref{three-by-three}, it is not OSP; therefore, by the same proposition in \cite{Li2015} that is used in \cref{three-by-three}, neither is $\impl$.
\end{proof}

As \cref{not-osp} shows, it is enough that \emph{some three women} have preferences that satisfy \cref{tricyclical} with respect to \emph{some three men} in order for obvious strategy-proofness to be unattainable. This implies that obvious strategy-proofness in also unattainable in large random markets with high probability.
\begin{corollary}[impossibility result for random markets]\label{not-osp-whp}
If $|M|\ge3$ and $|W|\ge3$, then as $|M|+|W|$ grows, we have for a randomly drawn preference profile $\bar{q}\sim U\bigl(\prefs(M)^W\bigr)$ for $W$ over $M$ that:\footnote{This result also holds, with the same proof, if $\bar{q}$ is drawn uniformly at random from the set of all full preferences (i.e., where each woman prefers being matched to any man over being unmatched).}
\begin{parts}
\item
With high probability no $\bar{q}$-stable (one-side-querying) matching rule is OSP-imple\-mentable.
\item
For every three distinct men $a,b,c\in M$, as $|W|$ grows, with high probability no $\bar{q}$-stable (one-side-querying) matching mechanism is OSP for $a$, $b$, and $c$.
\item\label{not-osp-whp-poly}
If $|M|\le\poly\bigl(|W|\bigr)$, then with high probability no $\bar{q}$-stable (one-side-querying) matching mechanism is OSP for more than two men.
\end{parts}
\end{corollary}
\cref{not-osp-whp} follows from an argument similar  to the one in the proof of \cref{not-osp}. Indeed, our proof of \cref{not-osp} in fact shows that if $\bar{q}$ satisfies \cref{tricyclical} with respect to three men $a,b,c$ and three women $x,y,z$, then no $\bar{q}$-stable matching mechanism is OSP for $a$, $b$, and~$c$. For \cref{not-osp-whp-poly}, for instance, we note that for a fixed triplet of distinct men $a,b,c\in M$, the probability that \cref{tricyclical} is not satisfied by $\bar{q}$ with respect to $a,b,c$ and any three women $x,y,z$ decreases exponentially with $|W|$, while the number of triplets of men increases polynomially with $|M|$.

We conclude this \lcnamecref{sec:impossibility} by noting that while the aesthetic preference profile defined in \cref{tricyclical} is sufficient for proving \cref{not-osp} and even \cref{not-osp-whp}, it is by no means the unique preference profile that eludes an obviously strategy-proof implementation, even when $|M|=|W|=3$. Indeed, \cref{not-osp-less-cyclical} in \cref{app:not-osp-less-cyclical} gives an additional example of such a preference profile, which could be described as ``less cyclical,'' in some sense.\footnote{While the proof of \cref{not-osp-less-cyclical} also follows a pruning argument, the reasoning is more involved than in the proof given for \cref{three-by-three} above.}
In this context, it is worth noting that following up on our paper, \cite{troyan2016} gives a necessary and sufficient condition, ``weak acyclicality'' (weaker, indeed, than acyclicality as defined in \cref{acyclicality}), on the preferences of objects in the (Pareto efficient, not necessarily stable) top trading cycles algorithm for this algorithm to be OSP-implementable for the agents. The example given in \cref{not-osp-less-cyclical} also demonstrates that \citeauthor{troyan2016}'s condition does not suffice for the existence of an OSP-implementable stable mechanism.
A comparison of the respective preference profiles used for the positive result of \cref{osp-more-cyclical} and the negative result of \cref{not-osp-less-cyclical},
noting that the former is obtained by taking the latter and arguably making it ``more aligned'' by modifying the preference list of woman $x$ to equal that of woman $y$,
suggests that an analogous succinct ``maximal domain'' characterization of preference profiles that admit OSP-implementable stable mechanisms may be delicate, and obtaining it may be challenging.

\section{Matching with two strategic sides}\label{sec:two-sided}

So far, this paper has studied two-sided matching markets in which only men are strategic and  women's preference lists  are commonly known. This allowed us to ask questions such as, for which preference profiles of women
one can OSP-implement the  $M$-optimal stable matching rule? This setting is furthermore practically relevant in school choice where, for example,   schools do not act strategically but have priorities over students.

Our analysis, however,  also  immediately yields that  when both men and women behave strategically, no stable matching mechanism is OSP-implementable. To formalize this result, we introduce a few definitions.
A \emph{two-sides-querying matching rule} is a function $C:\prefs(W)^M\times\prefs(M)^W\rightarrow\matchings$, from preference profiles for both men and women to a matching between $M$ and $W$.
A two-sides-querying matching rule $C$ is \emph{stable} if for any preference profiles $\bar{p}$ and $\bar{q}$ for men and women,  $C(\bar{p},\bar{q})$ is stable with respect to $\bar{p}$ and $\bar{q}$. A two-sides-querying matching mechanism\footnote{The definition of mechanisms and OSP for markets where both sides are strategic extends naturally from that given in \cref{osp} for markets where only one side is strategic.
The interested reader is referred to Appendix \ref{two-sided-mechanisms} for precise details.}
is \emph{stable} if the two-sides-querying matching rule that it implements is stable.
\cref{not-osp} implies the following impossibility result for two-sides-querying matching mechanisms:

\begin{corollary}[impossibility result for two-sides-querying mechanisms]\label{not-osp-two-sides}
If $|M|\ge3$ and $|W|\ge3$, then no stable two-sides-querying matching rule is OSP-implementable for $M$. Moreover, no stable two-sides-querying matching mechanism is OSP for more than two men.
\end{corollary}

As with \cref{not-osp}, we note that \cref{not-osp-two-sides} applies to any stable two-sides-querying matching rule, and not only to the \emph{$M$-optimal two-side-querying stable matching rule} (i.e., the two-sides-querying matching rule that maps each pair of preference profiles to the corresponding $M$-optimal stable matching).
Similarly, \cref{acyclical-osp} implies the following possibility result for two-sides-querying matching mechanisms:

\begin{corollary}[positive result for $|M|=2$ for two-sides-querying mechanisms]\label{osp-two-sides}
If $|M|=2$, then the two-sides-querying $M$-optimal stable matching rule is OSP-implementable (by first querying the women, and then, given their preferences, continuing as in \cref{acyclical-osp}).
\end{corollary}

\noindent
A precise argument that relates the results for markets with one strategic side and those for markets with two strategic sides is given  in \cref{App-two-sided-proof}.

\section{Discussion}
\label{sec:conclusion}

This paper finds that no stable matching mechanism is obviously strategy-proof for the participants even on one of the sides of the market. This suggests that there may not be any alternative  way to describe the deferred acceptance procedure that makes its strategy-proofness more apparent, implying that strategic mistakes observed in practice \citep{chen-sonmez2006,rees2016suboptimal,HRS2016,SS2017} may not be avoidable by better explaining the mechanism.  This highlights the importance of gaining the trust of the agents who participate in stable mechanisms,  so that they both act as advised (even when it is hard to verify that no strategic opportunities exist) and are assured that the social planner will not deviate from the prescribed  procedure after preferences are elicited.

For the case in which women's preferences are acyclical, we describe an OSP mechanism that implements the men-optimal stable matching. As may be expected, the strategy-proofness of this OSP implementation is easier to understand than that of deferred acceptance. It is interesting to compare and contrast this mechanism with OSP mechanisms for auctions. In binary allocation problems, such as private-value auctions with unit demand, procurement auctions with unit supply, and binary public good problems, \cite{Li2015} shows that in every OSP mechanism, each buyer chooses, roughly speaking, between a fixed option (i.e., quitting) and a ``moving'' option that is \emph{worsening} over time (i.e., its price is increasing). In contrast, in the OSP mechanism that we construct for the men-optimal stable matching with acyclical women's preferences, each man $m$  either is assigned his (current) top choice or chooses between a fixed option (i.e., being unmatched) and a ``moving'' option that is \emph{improving} over time: choosing any  woman who prefers $m$ most among all yet-to-be-matched men. This novel construction has come to be utilized by various OSP implementations, such as all of those that are surveyed in the end of the introduction.

Bridging the negative and positive results via an exact, succinct characterization of how aligned the preference profile of the proposed-to side needs to be in order to support an obviously strategy-proof implementation remains an open question. A comparison of the respective preference profiles used for the positive result of \cref{osp-more-cyclical} and the negative result of \cref{not-osp-less-cyclical} (in \cref{app:not-osp-less-cyclical}) suggests that such a succinct ``maximal domain'' characterization may be delicate, and obtaining it may be challenging.\footnote{While a technical challenge, we find it unlikely that resolving this problem will yield interesting economic insights.}

Interestingly, while deferred acceptance is weakly group strategy-proof and has an ascending flavor similar to that of ascending unit-demand auctions or clock auctions (which are all obviously strategy-proof), deferred acceptance is in fact not OSP-implementable.
It seems that the fact that stability is a two-sided objective (concerning the preferences of agents on both sides of the market), in contrast with maximizing efficiency or welfare for one side, increases the difficulty of employing strategic reasoning over stable mechanisms. In this context, it is worth noting a line of work \citep{Segal07,GNOR15} that highlights a similar message in terms of complexity rather than strategic reasoning, by showing that the communication complexity (measured in the number of messages) of finding, or even verifying, an approximately stable matching is significantly higher than the communication complexity of approximate welfare maximization for one of the sides of the market \citep{DNO14}.
Indeed, in more than one way, stability is not an ``obvious" objective.

While direct-revelation stable mechanisms are ubiquitous, there is  growing usage of  sequential-like implementations of deferred acceptance (or close variations thereof).\footnote{These include college admissions in Brazil \citep{bo2016iterative},  Inner Mongolia \citep{chen2015time,gong2016dynamic}, and Tunisia \citep{luflade2017value}, and school choice in Wake  County \citep{dur2018identifying}. These implementations  differ in various dimensions including  the type of information provided to students, the timing, and how students can revise their choices; such differences may very well impact the students' behavior and therefore the outcome.}
Our results imply that none of these variants, however presented to students and however conducted, can be OSP. Moreover,
when DA is implemented sequentially according to its traditional description, sincere behavior is no longer even a dominant strategy but only induces an ex-post equilibrium \citep{bo2016iterative}.\footnote{\citet{bo2016iterative} require that only rejected agents  may  revise their proposals at each step in order to eliminate possible manipulations that appeared in the mechanism for college admissions in Brazil.} Nonetheless, seemingly contrasting these theoretical results, experimental evidence shows that such a sequential implementation of DA leads more often to sincere behavior and stable outcomes than the static implementation \citep{bo2016iterative2,pais2017static}. While sequential-like implementations do~not possess stronger incentive properties than static implementations (and sometimes even possess weaker incentive properties), sequential-like implementations do ease the cognitive tasks of participants in various ways: they simplify strategic interactions by allowing students to break-down their decisions into smaller decisions that each requires somewhat less contingent reasoning than in the static implementation (as it is taken after receiving more information and feedback, such as the updated cutoff at each college); they allow students to focus on their next choice rather than to dwell on tentative choices that may never be reached; and they reduce the necessary preference communication and preference learning due to the information that is released throughout the mechanism \citep{bo2016iterative,ashlagi2017communication}.
Our findings formally demonstrate the cognitive complexity of reasoning in stable mechanisms; this suggests a possible explanation as to why, within the context of such mechanisms, the benefits from reducing cognitive load that are offered by sequential implementations outweigh the negative effects of the slightly weaker incentive properties of these implementations.\footnote{The impact on students' welfare (which is of major importance) is beyond the scope of this paper, but see, for example, \citet{luflade2017value,dur2018identifying}.} In a sense, in the absence of an OSP mechanism to reduce the cognitive load while strengthening the incentive properties, the ``next best thing'' may well be an extensive-form mechanism that eases cognitive load in different manners than OSP mechanisms, and moreover gives students a ``feeling'' similar to that of OSP mechanisms by not relinquishing control to the mechanism and by being able to constantly witness that the mechanism is run as promised.\footnote{For a recent definition of a very strong sense of witnessing that the mechanism is run as promised, see \cite{akbarpour2017credible}.}

\bibliographystyle{abbrvnat}
\bibliography{matching-not-obvious}

\begin{thebibliography}{40}
\providecommand{\natexlab}[1]{#1}
\providecommand{\url}[1]{\texttt{#1}}
\expandafter\ifx\csname urlstyle\endcsname\relax
  \providecommand{\doi}[1]{doi: #1}\else
  \providecommand{\doi}{doi: \begingroup \urlstyle{rm}\Url}\fi

\bibitem[Abdulkadiro{\u{g}}lu and S{\"{o}}nmez(2003)]{as2003}
A.~Abdulkadiro{\u{g}}lu and T.~S{\"{o}}nmez.
\newblock School choice: A mechanism design approach.
\newblock \emph{American Economic Review}, 93\penalty0 (3):\penalty0 729--747,
  2003.

\bibitem[Abdulkadiro{\u{g}}lu et~al.(2005)Abdulkadiro{\u{g}}lu, Pathak, Roth,
  and S{\"o}nmez]{aprs2005}
A.~Abdulkadiro{\u{g}}lu, P.~A. Pathak, A.~E. Roth, and T.~S{\"o}nmez.
\newblock The {B}oston public school match.
\newblock \emph{American Economic Review}, 95\penalty0 (2):\penalty0 368--371,
  2005.

\bibitem[Abdulkadiro{\u{g}}lu et~al.(2006)Abdulkadiro{\u{g}}lu, Pathak, Roth,
  and S{\"o}nmez]{abdulkadiroglu2006changing}
A.~Abdulkadiro{\u{g}}lu, P.~Pathak, A.~E. Roth, and T.~S{\"o}nmez.
\newblock Changing the {B}oston school choice mechanism.
\newblock Working paper 11965, National Bureau of Economic Research, 2006.

\bibitem[Abdulkadiro{\u{g}}lu et~al.(2009)Abdulkadiro{\u{g}}lu, Pathak, and
  Roth]{apr2009}
A.~Abdulkadiro{\u{g}}lu, P.~A. Pathak, and A.~E. Roth.
\newblock Strategy-proofness versus efficiency in matching with indifferences:
  Redesigning the {NYC} high school match.
\newblock \emph{American Economic Review}, 5\penalty0 (99):\penalty0
  1954--1978, 2009.

\bibitem[Akbarpour and Li(2017)]{akbarpour2017credible}
M.~Akbarpour and S.~Li.
\newblock Credible mechanisms.
\newblock Mimeo, 2017.

\bibitem[Ashlagi et~al.(2017{\natexlab{a}})Ashlagi, Braverman, Kanoria, and
  Shi]{ashlagi2017communication}
I.~Ashlagi, M.~Braverman, Y.~Kanoria, and P.~Shi.
\newblock Communication requirements and informative signaling in matching
  markets.
\newblock In \emph{Proceedings of the 18th ACM Conference on Economics and
  Computation (EC 2017)}, page 263, 2017{\natexlab{a}}.

\bibitem[Ashlagi et~al.(2017{\natexlab{b}})Ashlagi, Kanoria, and
  Leshno]{akl2013}
I.~Ashlagi, Y.~Kanoria, and J.~D. Leshno.
\newblock Unbalanced random matching markets: The stark effect of competition.
\newblock \emph{Journal of Political Economy}, 125\penalty0 (1):\penalty0
  69--98, 2017{\natexlab{b}}.

\bibitem[Bade and Gonczarowski(2017)]{bg2016}
S.~Bade and Y.~A. Gonczarowski.
\newblock Gibbard-{S}atterthwaite success stories and obvious
  strategyproofness.
\newblock In \emph{Proceedings of the 18th ACM Conference on Economics and
  Computation (EC 2017)}, page 565, 2017.

\bibitem[B{\'o} and Hakimov(2016{\natexlab{a}})]{bo2016iterative}
I.~B{\'o} and R.~Hakimov.
\newblock The iterative deferred acceptance mechanism.
\newblock Mimeo, 2016{\natexlab{a}}.

\bibitem[B{\'o} and Hakimov(2016{\natexlab{b}})]{bo2016iterative2}
I.~B{\'o} and R.~Hakimov.
\newblock Iterative versus standard deferred acceptance: Experimental evidence.
\newblock Mimeo, 2016{\natexlab{b}}.

\bibitem[Bogomolnaia et~al.(2005)Bogomolnaia, Deb, and Ehlers]{BDE05}
A.~Bogomolnaia, R.~Deb, and L.~Ehlers.
\newblock Strategy-proof assignment on the full preference domain.
\newblock \emph{Journal of Economic Theory}, 123\penalty0 (2):\penalty0
  161--186, 2005.

\bibitem[Chen and Pereyra(2015)]{chen2015time}
L.~Chen and J.~Pereyra.
\newblock Time-constrained school choice.
\newblock Mimeo, 2015.

\bibitem[Chen et~al.(2016)Chen, Egesdal, Pycia, and
  Yenmez]{chen2016manipulability}
P.~Chen, M.~Egesdal, M.~Pycia, and M.~B. Yenmez.
\newblock Manipulability of stable mechanisms.
\newblock \emph{American Economic Journal: Microeconomics}, 8\penalty0
  (2):\penalty0 202--214, 2016.

\bibitem[Chen and S{\"o}nmez(2006)]{chen-sonmez2006}
Y.~Chen and T.~S{\"o}nmez.
\newblock School choice: an experimental study.
\newblock \emph{Journal of Economic Theory}, 127\penalty0 (1):\penalty0
  202--231, 2006.

\bibitem[Demange et~al.(1986)Demange, Gale, and Sotomayor]{demange1986multi}
G.~Demange, D.~Gale, and M.~Sotomayor.
\newblock Multi-item auctions.
\newblock \emph{Journal of Political Economy}, 94\penalty0 (4):\penalty0
  863--872, 1986.

\bibitem[Dobzinski et~al.(2014)Dobzinski, Nisan, and Oren]{DNO14}
S.~Dobzinski, N.~Nisan, and S.~Oren.
\newblock Economic efficiency requires interaction.
\newblock In \emph{Proceedings of the 46th Annual ACM Symposium on Theory of
  Computing (STOC 2014)}, pages 233--242, 2014.
\newblock Full version available at arXiv:1311.4721.

\bibitem[Dubins and Freedman(1981)]{Dubins-Freedman}
L.~E. Dubins and D.~A. Freedman.
\newblock Machiavelli and the {G}ale-{S}hapley algorithm.
\newblock \emph{American Mathematical Monthly}, 88\penalty0 (7):\penalty0
  485--494, 1981.

\bibitem[Dur et~al.(2018)Dur, Hammond, and Morrill]{dur2018identifying}
U.~Dur, R.~G. Hammond, and T.~Morrill.
\newblock Identifying the harm of manipulable school-choice mechanisms.
\newblock \emph{American Economic Journal: Economic Policy}, 10\penalty0
  (1):\penalty0 187--213, 2018.

\bibitem[Ergin(2002)]{Ergin2002}
H.~I. Ergin.
\newblock Efficient resource allocation on the basis of priorities.
\newblock \emph{Econometrica}, 70\penalty0 (6):\penalty0 2489--2497, 2002.

\bibitem[Gale and Shapley(1962)]{Gale-Shapley}
D.~Gale and L.~S. Shapley.
\newblock College admissions and the stability of marriage.
\newblock \emph{American Mathematical Monthly}, 69\penalty0 (1):\penalty0
  9--15, 1962.

\bibitem[Gale and Sotomayor(1985)]{Gale-Sotomayor-1985}
D.~Gale and M.~Sotomayor.
\newblock Ms. {M}achiavelli and the stable matching problem.
\newblock \emph{American Mathematical Monthly}, 92\penalty0 (4):\penalty0
  261--268, 1985.

\bibitem[Gonczarowski et~al.(2015)Gonczarowski, Nisan, Ostrovsky, and
  Rosenbaum]{GNOR15}
Y.~A. Gonczarowski, N.~Nisan, R.~Ostrovsky, and W.~Rosenbaum.
\newblock A stable marriage requires communication.
\newblock In \emph{Proceedings of the 26th Annual ACM-SIAM Symposium on
  Discrete Algorithms (SODA 2015)}, pages 1003--1017, 2015.

\bibitem[Gong and Liang(2016)]{gong2016dynamic}
B.~Gong and Y.~Liang.
\newblock A dynamic college admission mechanism in {I}nner {M}ongolia: Theory
  and experiment.
\newblock Mimeo, 2016.

\bibitem[Hassidim et~al.(2016)Hassidim, Romm, and Shorrer]{HRS2016}
A.~Hassidim, A.~Romm, and R.~I. Shorrer.
\newblock `{S}trategic' behavior in a strategy-proof environment.
\newblock Mimeo, 2016.

\bibitem[Hassidim et~al.(2017)Hassidim, Marciano, Romm, and Shorrer]{HMRS2017}
A.~Hassidim, D.~Marciano, A.~Romm, and R.~I. Shorrer.
\newblock The mechanism is truthful, why aren't you?
\newblock \emph{American Economic Review}, 107\penalty0 (5):\penalty0 220--224,
  2017.

\bibitem[Immorlica and Mahdian(2005)]{immorlica2005marriage}
N.~Immorlica and M.~Mahdian.
\newblock Marriage, honesty, and stability.
\newblock In \emph{Proceedings of the 16th Annual ACM-SIAM Symposium on
  Discrete Algorithms (SODA 2005)}, pages 53--62, 2005.

\bibitem[Kagel et~al.(1987)Kagel, Harstad, and Levin]{kagel1987information}
J.~H. Kagel, R.~M. Harstad, and D.~Levin.
\newblock Information impact and allocation rules in auctions with affiliated
  private values: A laboratory study.
\newblock \emph{Econometrica}, 55\penalty0 (6):\penalty0 1275--1304, 1987.

\bibitem[Kojima and Pathak(2009)]{kojima2009incentives}
F.~Kojima and P.~A. Pathak.
\newblock Incentives and stability in large two-sided matching markets.
\newblock \emph{American Economic Review}, 99\penalty0 (3):\penalty0 608--627,
  2009.

\bibitem[Li(2017)]{Li2015}
S.~Li.
\newblock Obviously strategy-proof mechanisms.
\newblock \emph{American Economic Review}, 107\penalty0 (11):\penalty0
  3257--3287, 2017.

\bibitem[Luflade(2017)]{luflade2017value}
M.~Luflade.
\newblock The value of information in centralized school choice systems.
\newblock Mimeo (job market paper), 2017.

\bibitem[Milgrom and Segal(2014)]{milgrom2014deferred}
P.~Milgrom and I.~Segal.
\newblock Deferred-acceptance auctions and radio spectrum reallocation.
\newblock In \emph{Proceedings of the 15th ACM Conference on Economics and
  Computation (EC 2014)}, pages 185--186, 2014.

\bibitem[Pais et~al.(2016)Pais, Klijn, and Vorsatz]{pais2017static}
J.~Pais, F.~Klijn, and M.~Vorsatz.
\newblock Static versus dynamic deferred acceptance in school choice: Theory
  and experiment.
\newblock Working Paper 926, Barcelona GSE, 2016.

\bibitem[Pathak and S{\"o}nmez(2008)]{pathak2008leveling}
P.~A. Pathak and T.~S{\"o}nmez.
\newblock Leveling the playing field: Sincere and sophisticated players in the
  {B}oston mechanism.
\newblock \emph{American Economic Review}, 98\penalty0 (4):\penalty0
  1636--1652, 2008.

\bibitem[Pycia and Troyan(2016)]{pt2016}
M.~Pycia and P.~Troyan.
\newblock Obvious dominance and random priority.
\newblock Mimeo, 2016.

\bibitem[Rees-Jones(2016)]{rees2016suboptimal}
A.~Rees-Jones.
\newblock Suboptimal behavior in strategy-proof mechanisms: Evidence from the
  residency match.
\newblock Mimeo, 2016.

\bibitem[Roth(1984)]{roth1984evolution}
A.~E. Roth.
\newblock The evolution of the labor market for medical interns and residents:
  a case study in game theory.
\newblock \emph{Journal of Political Economy}, 92\penalty0 (6):\penalty0
  991--1016, 1984.

\bibitem[Roth(2002)]{Roth2002-economist-engineer}
A.~E. Roth.
\newblock The economist as engineer: Game theory, experimentation and
  computation as tools for design economics.
\newblock \emph{Econometrica}, 70\penalty0 (4):\penalty0 1341--1378, 2002.

\bibitem[Segal(2007)]{Segal07}
I.~Segal.
\newblock The communication requirements of social choice rules and supporting
  budget sets.
\newblock \emph{Journal of Economic Theory}, 136\penalty0 (1):\penalty0
  341--378, 2007.

\bibitem[Shorrer and S{\'o}v{\'a}g{\'o}(2017)]{SS2017}
R.~I. Shorrer and S.~S{\'o}v{\'a}g{\'o}.
\newblock Obvious mistakes in a strategically simple college-admissions
  environment.
\newblock Mimeo, 2017.

\bibitem[Troyan(2016)]{troyan2016}
P.~Troyan.
\newblock Obviously strategyproof implementation of allocation mechanisms.
\newblock Mimeo, 2016.

\end{thebibliography}

\appendix

\section{Mechanisms with restricted domains}\label{restricted-domain-mechanisms}

In this \lcnamecref{restricted-domain-mechanisms}, we explicitly adapt  the definitions in \cref{osp} to a restricted domain of preferences, as used in the proof of \cref{three-by-three}. The differences from the definitions in \cref{osp} are marked with an \uline{underscore}. We emphasize that these definitions, like those in \cref{osp}, are also a special case of the definitions in \cite{Li2015}. For every $m\in M$, fix a subset $\prefs_m\subseteq\prefs(W)$. Furthermore, define $\prefs\eqdef\bigtimes_{m\in M}\prefs_m$.

\begin{definition}[matching mechanism]
A (one-side-querying extensive-form) \emph{matching mechanism} for $M$ over $W$ \uline{with respect to $\prefs$} consists of:
\begin{enumerate}
\item
A rooted tree $T$.
\item
A map $X:L(T)\rightarrow\matchings(M,W)$ from the leaves of $T$ to matchings between $M$ and $W$.
\item
A map $Q:V(T)\setminus L(T)\rightarrow M$, from internal nodes of $T$ to $M$.
\item
A map $A:E(T)\rightarrow 2^{\prefs(W)}$, from edges of $T$ to predicates over $\prefs(W)$, such that all of the following hold:
\begin{itemize}
\item
Each such predicate must match at least one element in $\prefs(W)$.
\item
The predicates corresponding to edges outgoing from the same node are disjoint.
\item
The disjunction (i.e., set union) of all predicates corresponding to edges outgoing from a node $n$ equals the predicate corresponding to the last edge outgoing from a node labeled $Q(n)$ along the path from the root to $n$, or to the predicate matching all elements of $\uline{\prefs_{Q(n)}}$ if no such edge exists.\footnote{In particular, this implies that the predicates corresponding to edges outgoing from a node $n$ are predicates over $\uline{\prefs_{Q(n)}}$.}
\end{itemize}
\end{enumerate}
\end{definition}

A preference profile $\bar{p}\in\uline{\prefs}$ is said to \emph{pass through} a node $n \in V(T)$ if, for each edge~$e$ along the path from the root of $T$ to $n$, it is the case that $p_{Q(n')}\in A(e)$, where $n'$ is the source node of $e$. That is, the nodes through which $\bar{p}$ passes are the nodes of the path that starts from the root of $T$ and follows, from each internal node $n'$ that it reaches, the unique outgoing edge whose predicate matches the preference list of $Q(n')$.

\begin{definition}[implemented matching rule]
Given an extensive-form matching mechanism~$\impl$ \uline{with respect to $\prefs$}, we denote by $C^{\impl}$, called the matching rule \emph{implemented by} $\impl$, the (one-side-querying) matching rule mapping a preference profile $\bar{p}\in\uline{\prefs}$ to the matching $X(n)$, where $n$ is the unique leaf through which $\bar{p}$ passes. Equivalently, $n$ is the node in $T$ obtained by traversing $T$ from its root, and from each internal node $n'$ that is reached, following the unique outgoing edge whose predicate matches the preference list of $Q(n')$.
\end{definition}

Two preference lists $p,p'\in\prefs(W)$ are said to \emph{diverge} at a node $n\in V(T)$ if there exist two distinct edges~$e,e'$ outgoing from $n$ such that  $p\in A(e)$ and $p'\in A(e')$.\footnote{In particular, this implies that $p,p'\in\uline{\prefs_{Q(n)}}$.}

\begin{definition}[obvious strategy-proofness (OSP)]
Let $\impl$ be an extensive-form matching mechanism \uline{with respect to $\prefs$} .
\begin{enumerate}
\item
$\impl$ is said to be \emph{obviously strategy-proof (OSP) for a man $m\in M$} if for every node $n$ with $Q(n)=m$ and for every
$\bar{p}=(p_{m'})_{m'\in M}\in\uline{\prefs}$ and $\bar{p}'=(p'_{m'})_{m'\in M}\in\uline{\prefs}$ that both pass through $n$ such that $p_m$ and $p'_m$ diverge at $n$, it is the case that $C^{\impl}_m(\bar{p}) \succeq_m C^{\impl}_m(\bar{p}')$ according to~$p_m$. In other words, the worst possible outcome for $m$ when acting truthfully (i.e., according to $p_m$) at $n$ is no worse than the best possible outcome for $m$ when misrepresenting his preference list to be  $p'_m$ at $n$.
\item
$\impl$ is said to be \emph{obviously strategy-proof (OSP)} if it is obviously strategy-proof for every man $m\in M$.
\end{enumerate}
\end{definition}

\section{A \texorpdfstring{``}{"}less cyclical\texorpdfstring{''}{"} non-OSP-implementable example}\label{app:not-osp-less-cyclical}

In this \lcnamecref{app:not-osp-less-cyclical}, we give an additional example of a preference profile $\bar{q}\in\prefs(M)^W$,
for three women over three men, for which no $\bar{q}$-stable matching rule is OSP-implementable.
This preference profile could be described, in some sense, as ``less cyclical'' than the one used above to drive the proof of the results of \cref{sec:impossibility}.
(Indeed, as noted above, this non-OSP-implementable preference profile is obtained by taking the OSP-implementable preference profile from \cref{osp-more-cyclical} and arguably making it ``more aligned'' by modifying the preference list of woman $x$ to equal that of woman $y$.)
While, similarly to the proof of \cref{three-by-three}, we show the impossibility of OSP-implementation of this example via a pruning argument, the reasoning in this argument is more involved than in the one in the proof given for \cref{three-by-three} in \cref{sec:impossibility}.

\begin{proposition}\label{not-osp-less-cyclical}
For $|M|=|W|=3$, no OSP mechanism implements a $\bar{q}$-stable (one-side-querying) matching rule, for the following preference profile $\bar{q}\in\prefs(M)^W$ for $M$ over $W$ (where each woman prefers being matched to any man over being unmatched):
\[
\begin{array}{ccccc}
a & \succ_x & c & \succ_x&  b \\
a & \succ_y & c & \succ_y & b \\
b & \succ_z & a & \succ_z & c.
\end{array}
\]
\end{proposition}

\begin{proof}
The proof starts similarly to that of \cref{three-by-three}.
Let $M=\{a,b,c\}$ and $W=\{x,y,z\}$. Let $\bar{q}$ be the above preference profile, and assume for contradiction that an OSP mechanism~$\impl$ that implements a $\bar{q}$-stable matching rule $C^{\impl}$ exists. Therefore, $C^{\impl}$ is strategy-proof, and so, by \cref{strategy-proof}, $C^{\impl}=C^{\bar{q}}$.
In order to reach a contradiction we dramatically restrict the domain of preferences of all men, however in this proof to a slightly richer domain than in the proof of \cref{three-by-three}. We define:

\begin{center}
\setlength\tabcolsep{2em}
\begin{tabular}{ccc}
$p^1_a \eqdef z \succ x \succ y$ & $p^1_b \eqdef y \succ z \succ x$ & $p^1_c \eqdef x \succ y \succ z$ \\
$p^2_a \eqdef z \succ y \succ x$ & $p^2_b \eqdef x \succ z \succ y$ & $p^2_c \eqdef y \succ x \succ z$, \\
& $p^3_b \eqdef x \succ y \succ z $ &
\end{tabular}
\end{center}
and set $\prefs_a\eqdef\{p^1_a,p^2_a\}$, $\prefs_b\eqdef\{p^1_b,p^2_b,p^3_b\}$, and $\prefs_c\eqdef\{p^1_c,p^2_c\}$.

Following a proof technique in \cite{Li2015}, we prune (see the proof of \cref{three-by-three} for more details) the tree of $\impl$ according to $\prefs_a,\prefs_b,\prefs_c$, to obtain a mechanism that implements~$C^{\bar{q}}$ where the preference list of each man $m\in M$ is \emph{a priori} restricted to be in $\prefs_m$.
By a proposition in \cite{Li2015}, since the original mechanism $\impl$ is OSP, so is the pruned mechanism as well.

Let $n$ be the earliest (i.e., closest to the root) node in the pruned tree that has more than one outgoing edge (such a node clearly exists, since $C^{\impl}=C^{\bar{q}}$ is not constant over $\prefs_a\times\prefs_b\times\prefs_c$). While the lack of symmetry of $\bar{q}$ does requires a slightly longer argument compared to the proof of \cref{three-by-three} to complete this proof (reasoning by cases according to~$Q(n)$ below), what makes the reasoning in this argument more involved (see the reasoning in the case $Q(n)=b$ below) than in its counterpart in the proof of \cref{three-by-three} is the fact that we have left possible three preference lists for man $b$.\footnote{To our knowledge, the first instance of an impossibility-by-pruning proof with more than two possible preference lists/types for any of the agents is in an impossibility result for OSP-implementation of combinatorial auctions in \cite{bg2016}. While that paper is much newer than any other result in our paper, the first draft of that proof predated the proof given in this \lcnamecref{app:not-osp-less-cyclical}.} We conclude the proof by reasoning by cases according to the identity of $Q(n)$, in each case obtaining a contradiction by showing that the pruned tree is in fact not OSP.

\begin{enumerate}
\item[$Q(n)=a$]
By definition of pruning, it must be the case that $n$ has two outgoing edges, one labeled $p^1_a$, and the other labeled $p^2_a$.
In this case, for $p_a=p^1_a$ (the ``true preferences''), $p_b=p^1_b$, and $p_c=p^2_c$, we have that $C_a^{\impl}(\bar{p})=C_a^{\bar{q}}(\bar{p})=x$, yet for $p'_a=p^2_a$ (a ``possible manipulation''), $p'_b=p^2_b$, and $p'_c=p^2_c$, we have that $C_a^{\impl}(\bar{p}')=C_a^{\bar{q}}(\bar{p}')=z$, even though $C_a^{\impl}(\bar{p}')=z\succ_a x = C_a^{\impl}(\bar{p})$ according to $p_a$ (by definition of $n$, both $\bar{p}$ and $\bar{p}'$ pass through~$n$, and $p_a$ and $p_a'$ diverge at $n$), and so the mechanism of the pruned tree indeed is not OSP --- a contradiction.
\item[$Q(n)=c$]
By definition of pruning, it must be the case that $n$ has two outgoing edges, one labeled $p^1_c$, and the other labeled $p^2_c$.
In this case, for $p_c=p^1_c$ (the ``true preferences''), $p_a=p^1_a$, and $p_b=p^2_b$, we have that $C_c^{\impl}(\bar{p})=C_c^{\bar{q}}(\bar{p})=y$, yet for $p'_c=p^2_c$ (a ``possible manipulation''), $p'_a=p^2_a$, and $p'_b=p^1_b$, we have that $C_c^{\impl}(\bar{p}')=C_c^{\bar{q}}(\bar{p}')=x$, even though $C_c^{\impl}(\bar{p}')=x\succ_c y = C_c^{\impl}(\bar{p})$ according to $p_c$ (by definition of $n$, both $\bar{p}$ and $\bar{p}'$ pass through~$n$, and $p_c$ and $p_c'$ diverge at $n$), and so the mechanism of the pruned tree indeed is not OSP --- a contradiction.
\item[$Q(n)=b$]
By definition of pruning, it must be the case that $n$ has at least two outgoing edges, and therefore has at least one edge labeled by a singleton preference list $p^i_b$. We prove this case by reasoning by subcases according to the value of $i$.
\begin{enumerate}
\item[i=1]
In this case, for $p_b=p^i_b=p^1_b$ (the ``true preferences''), $p_a=p^1_a$, and $p_c=p^2_c$, we have that $C_b^{\impl}(\bar{p})=C_b^{\bar{q}}(\bar{p})=z$, yet for $p'_b=p^3_b$ (a ``possible manipulation''), $p'_a=p^1_a$, and $p'_c=p^1_c$, we have that $C_b^{\impl}(\bar{p}')=C_b^{\bar{q}}(\bar{p}')=y$, even though $C_b^{\impl}(\bar{p}')=y\succ_b z = C_b^{\impl}(\bar{p})$ according to $p_b$ (by definition of $n$, both $\bar{p}$ and $\bar{p}'$ pass through~$n$, and since $i=1$ we have that $p_b=p^i_b$ and $p_b'\ne p^i_b$ diverge at $n$), and so the mechanism of the pruned tree indeed is not OSP --- a contradiction.
\item[i=2]
In this case, for $p_b=p^i_b=p^2_b$ (the ``true preferences''), $p_a=p^2_a$, and $p_c=p^1_c$, we have that $C_b^{\impl}(\bar{p})=C_b^{\bar{q}}(\bar{p})=z$, yet for $p'_b=p^3_b$ (a ``possible manipulation''), $p'_a=p^1_a$, and $p'_c=p^2_c$, we have that $C_b^{\impl}(\bar{p}')=C_b^{\bar{q}}(\bar{p}')=x$, even though $C_b^{\impl}(\bar{p}')=x\succ_b z = C_b^{\impl}(\bar{p})$ according to $p_b$ (by definition of $n$, both $\bar{p}$ and $\bar{p}'$ pass through~$n$, and since $i=2$ we have that $p_b=p^i_b$ and $p_b'\ne p^i_b$ diverge at $n$), and so the mechanism of the pruned tree indeed is not OSP --- a contradiction.
\item[i=3]
In this case, for $p_b=p^i_b=p^3_b$ (the ``true preferences''), $p_a=p^1_a$, and $p_c=p^1_c$, we have that $C_b^{\impl}(\bar{p})=C_b^{\bar{q}}(\bar{p})=y$, yet for $p'_b=p^2_b$ (a ``possible manipulation''), $p'_a=p^1_a$, and $p'_c=p^2_c$, we have that $C_b^{\impl}(\bar{p}')=C_b^{\bar{q}}(\bar{p}')=x$, even though $C_b^{\impl}(\bar{p}')=x\succ_b y = C_b^{\impl}(\bar{p})$ according to $p_b$ (by definition of $n$, both $\bar{p}$ and $\bar{p}'$ pass through~$n$, and since $i=3$ we have that $p_b=p^i_b$ and $p_b'\ne p^i_b$ diverge at $n$), and so the mechanism of the pruned tree indeed is not OSP --- a contradiction.\qedhere
\end{enumerate}
\end{enumerate}

\end{proof}

\section{Two-sides-querying mechanisms}\label{two-sided-mechanisms}

In this \lcnamecref{two-sided-mechanisms}, we explicitly  adapt  the definitions in \cref{osp} for two-sides-querying mechanisms, where the (strategic) participants include not only the men but also the women, as in \cref{sec:two-sided}. The differences from the definitions in \cref{osp} are marked with an \uline{underscore}. We emphasize that these definitions, like those in \cref{osp}, are also a special case of the definitions in \cite{Li2015}. Define $\prefs\eqdef\prefs(W)^M\times\prefs(M)^W$. For every two-sided preference profile $\bar{r}=(\bar{p},\bar{q})\in\prefs$, we write $r_m=p_m$ for every $m\in M$ and $r_w=q_w$ for every $w\in W$.

\begin{definition}[two-sides-querying matching mechanism]
A \uline{\emph{two-sides-querying}} (extensive-form) \emph{matching mechanism} for $M$ \uline{and} $W$ consists of:
\begin{enumerate}
\item
A rooted tree $T$.
\item
A map $X:L(T)\rightarrow\matchings(M,W)$ from the leaves of $T$ to matchings between $M$ and $W$.
\item
A map $Q:V(T)\setminus L(T)\rightarrow M\uline{\,\cup\,W}$, from internal nodes of $T$ to participants $M\uline{\,\cup\,W}$.
\item
A map $A:E(T)\rightarrow 2^{\prefs(W)}\uline{\,\cup\,2^{\prefs(M)}}$, from edges of $T$ to predicates over $\prefs(W)$ \uline{or over $\prefs(M)$}, such that  all of the following hold:
\begin{itemize}
\item
Each such predicate must match at least one element in $\prefs(W)$ if $Q(n)\in M$ \uline{and at least one element in $\prefs(M)$ if $Q(n)\in W$}.
\item
The predicates corresponding to edges outgoing from the same node are disjoint.
\item
The disjunction (i.e., set union) of all predicates corresponding to edges outgoing from a node $n$ equals the predicate corresponding to the last edge outgoing from a node labeled $Q(n)$ along the path from the root to $n$, or, if no such edge exists, to the predicate matching all elements of $\prefs(W)$ if $Q(n)\in M$ \uline{and all elements of $\prefs(M)$ if $Q(n)\in W$}.\footnote{In particular, this implies that the predicates corresponding to edges outgoing from a node $n$ are predicates over $\prefs(W)$ if $Q(n)\in M$ \uline{and over $\prefs(M)$ if $Q(n)\in W$}.}

\end{itemize}
\end{enumerate}
\end{definition}

A \uline{two-sides-querying} preference profile $\bar{r}\in\uline{\prefs}$ is said to \emph{pass through} a node $n \in V(T)$ if, for each edge~$e$ along the path from the root of $T$ to $n$, it is the case that $r_{Q(n')}\in A(e)$, where $n'$ is the source node of $e$. That is, the nodes through which $\bar{r}$ passes are the nodes of the path that starts from the root of $T$ and follows, from each internal node $n'$ that it reaches, the unique outgoing edge whose predicate matches the preference list of $Q(n')$.

\begin{definition}[implemented matching rule]
Given a \uline{two-sides-querying} extensive-form matching mechanism~$\impl$, we denote by $C^{\impl}$, called the \uline{two-sides-querying} matching rule \emph{implemented by} $\impl$, the \uline{two-sides-querying} matching rule  mapping a \uline{two-sides-querying} preference profile $\bar{r}\in\uline{\prefs}$ to the matching $X(n)$, where $n$ is the unique leaf through which $\bar{r}$ passes. Equivalently, $n$ is the node in $T$ obtained by traversing $T$ from its root, and from each internal node $n'$ that is reached, following the unique outgoing edge whose predicate matches the preference list of $Q(n')$.
\end{definition}

Two preference lists $r,r'\in\prefs(W)\uline{\,\cup\,\prefs(M)}$ are said to \emph{diverge} at a node $n\in V(T)$ if there exist two distinct edges~$e,e'$ outgoing from $n$ such that $r\in A(e)$ and $r'\in A(e')$.\footnote{In particular, this implies that $r,r'\in\prefs(W)$ if $Q(n)\in M$ \uline{and that $r,r'\in\prefs(M)$ if $Q(n)\in W$}.}

\begin{definition}[obvious strategy-proofness (OSP)]
Let $\impl$ be a \uline{two-sides-querying} extensive-form matching mechanism.
$\impl$ is said to be \emph{obviously strategy-proof (OSP) for a participant $a\in M\uline{\,\cup\,W}$} if for every node $n$ with $Q(n)=a$ and for every
$\bar{r},\bar{r}'\in\uline{\prefs}$ that both pass through $n$ such that  $p_a$ and $p'_a$ diverge at $n$, it is the case that $C^{\impl}_a(\bar{r}) \succeq_a C^{\impl}_a(\bar{r}')$ according to~$r_a$. In other words, the worst possible outcome for $a$ when acting truthfully (i.e., according to $r_a$) at $n$ is no worse than the best possible outcome for $a$ when misrepresenting his or her preference list to be $r'_a$ at $n$.
\end{definition}

\begin{definition}[OSP-implementability]
A \uline{two-sides-querying} matching rule $C:\uline{\prefs}\rightarrow\matchings(M,W)$ is said to be \emph{OSP-implementable} for a set of participants $A\subseteq M\uline{\,\cup\,W}$ if $C=C^{\impl}$ for some \uline{two-sides-querying} matching mechanism~$\impl$ that is OSP for (every participant in)~$A$.
\end{definition}

\section{From one strategic side to two strategic sides}
\label{App-two-sided-proof}

The next \lcnamecref{one-vs-two} allows us to  obtain results in the two-strategic-sides model from the results obtained in the one-strategic-side model (as alluded to in the discussion opening \cref{sec:two-sided}, the converse is not as immediate, e.g., neither \cref{acyclical-osp} nor \cref{not-osp-whp} is an immediate corollary of results that are naturally stated for two-sides-querying mechanisms/matching rules). Indeed, \cref{not-osp-two-sides,osp-two-sides} both follow via this \lcnamecref{one-vs-two} from the respective analogous results for one-side-querying mechanisms/matching rules.

\begin{lemma}[relation between one-side-querying and two-sides-querying OSP mechanisms]\label{one-vs-two}
For every $M'\subseteq M$, there exists a stable two-sides-querying matching mechanism that is OSP for $M'$ if and only if for every $\bar{q}\in\prefs(W)^M$ there exists a $\bar{q}$-stable one-side-querying matching mechanism that is OSP for $M'$.
\end{lemma}

\begin{proof}
$\Rightarrow$: Assume that there exists a stable two-sides-querying matching mechanism~$\impl$ that is OSP for $M'$, and let $\bar{q}\in\prefs(W)^M$. We prune (see the proof of \cref{three-by-three} for an explanation of pruning) the tree of $\impl$ such that  the women's preference profile is fixed to be $\bar{q}$. The resulting (pruned) mechanism is a \emph{one-side-querying} matching mechanism that is $\bar{q}$-stable and (by the same proposition in \cite{Li2015} that is used in \cref{three-by-three}) OSP for $M'$, as required.

$\Leftarrow$: Assume that for every $\bar{q}\in\prefs(M)^W$ there exists a $\bar{q}$-stable one-side-querying matching mechanism $\impl^{\bar{q}}$ that is OSP for $M'$. We construct a stable \emph{two-sides-querying} matching mechanism $\impl$ as follows: first ask all women, in some order, for all of their preference lists; the leaves of the tree so far are thus in one-to-one correspondence with preference profiles $\bar{q}\in\prefs(M)^W$ that pass through them. Next, at each ``interim leaf'' $n^{\bar{q}}$ corresponding to a preference profile $\bar{q}\in\prefs(M)^W$ (that passes through it), construct a subtree that is identical to the tree of $\impl^{\bar{q}}$, with $n^{\bar{q}}$ as its root. It is straightforward to verify that the fact that each $\impl^{\bar{q}}$ is $\bar{q}$-stable and OSP for $M'$ implies that $\impl$ is stable and OSP for $M'$.
\end{proof}

\end{document}